\pdfoutput=1
\documentclass[10pt,pra,aps,superscriptaddress,nofootinbib,twocolumn,longbibliography]{revtex4-1}
\usepackage{amsmath,bbm}
\usepackage{amsthm}
\usepackage{amsmath}
\usepackage{latexsym}
\usepackage{amssymb}
\usepackage{float}
\usepackage{quantikz}
\usepackage{tikz}
\usepackage{graphicx}           
\usepackage{color}
\usepackage{xcolor}
\usepackage{mathpazo}
\usepackage{comment}
\usepackage{enumitem}
\usepackage{multirow}
\usepackage{setspace}
\usepackage{subcaption}
\usepackage[colorlinks=true,linkcolor=blue,citecolor=magenta,urlcolor=blue]{hyperref}
\usepackage{svg}

\newcommand{\be}{\begin{equation}}
\newcommand{\ee}{\end{equation}}
\newcommand{\bea}{\begin{eqnarray}}
\newcommand{\eea}{\end{eqnarray}}
\def\squareforqed{\hbox{\rlap{$\sqcap$}$\sqcup$}}
\def\qed{\ifmmode\squareforqed\else{\unskip\nobreak\hfil
\penalty50\hskip1em\null\nobreak\hfil\squareforqed
\parfillskip=0pt\finalhyphendemerits=0\endgraf}\fi}
\def\endenv{\ifmmode\;\else{\unskip\nobreak\hfil
\penalty50\hskip1em\null\nobreak\hfil\;
\parfillskip=0pt\finalhyphendemerits=0\endgraf}\fi}

\makeatletter
\newtheorem*{rep@theorem}{\rep@title}
\newcommand{\newreptheorem}[2]{%
\newenvironment{rep#1}[1]{%
 \def\rep@title{#2 \ref{##1}}%
 \begin{rep@theorem}}%
 {\end{rep@theorem}}}
\makeatother

\newtheorem{thm}{Theorem}
\newreptheorem{thm}{Theorem}
\newtheorem{lemma}{Lemma}
\newtheorem{definition}{Definition}

\newtheorem{obs}{Observation}

\newtheorem{coro}{Corollary}

\newtheorem{example}{Example}

\begin{document}

\title{Antidistinguishability of states in General Probabilistic Theories}

\author{Satyaki Manna}
\email{mannasatyaki@gmail.com}
\affiliation{Department of Physics, School of Basic Sciences, Indian Institute of Technology Bhubaneswar, Odisha 752050, India}
\author{Anandamay Das Bhowmik}
\email{ananda.adb@gmail.com}
\affiliation{S. N. Bose National Centre for Basic Sciences, Block JD, Sector III, Salt Lake, Kolkata 700 106, India}
\begin{abstract}
We investigate antidistinguishability of states within the framework of general probabilistic theories (GPTs). We formulate antidistinguishability, strong and equal antidistinguishability as refined notions that imposed additional constraint on the measurement effects. We establish general results relating these notions of antidistinguishability and derive an upper bound on the cardinality of equally antidistinguishable sets in terms of the affine dimension of the state space. We then study antidistinguishability in polygonal theories, obtaining conditions for antidistinguishability of a set of states. In consequence, we show that the set of all pure states in a polygon model is antidistinguishable. Additionally, we identify broad families of strongly and equally antidistinguishable states. Finally, using Random Exclusion Codes, whose success probability is governed by the antidistinguishability of different sets of encoding states, we probe the nonclassicality of polygon theories. We find that certain polygon models can outperform the optimal quantum value, while their optimal performance converges to the quantum limit in the large-polygon limit.
\end{abstract}

\maketitle
\section{Introduction}
Determining whether two physical processes can be reliably distinguished is a fundamental question in the characterization of physical theories. A related but less stringent concept is \emph{antidistinguishability}, which concerns ruling out possible processes rather than directly identifying the one that has occurred. More precisely, whereas distinguishability associates measurement outcomes with the identification of the underlying process \cite{Bae_2015,JánosABergou_2007,fuchs1996,manna25,manna26}, antidistinguishability permits certain alternative to be conclusively excluded from consideration \cite{Caves02,Johnston2025tightbounds,manna2025,manna20,Manna_2026,Russo}. Both notions have played an important role in quantum foundations, including studies concerning the ontological status and reality of quantum states \cite{leifer2014,barrett,Chaturvedi2020quantum,bhowmik2022,pbr,ray2024,chaturvedi26}. They also constitute important resources in quantum information theory, with applications to information processing and communication tasks \cite{PhysRevResearch.2.013326,leifer,srikumar,uola,Heinosaari_2019,manna_dist,manna2026,pandit}.

General probabilistic theories (GPTs) provide an operational framework that encompasses both classical and quantum theory while allowing for more general state spaces and measurement structures \cite{hardy2001,hardy2012limited,barrett04,PLAVALA20231,chiribella,chiribella10,janotta2014generalized,chiribella20}. Within this framework, physical states are represented by points in a convex state space, while measurements are described by effects acting on these states. This convex-geometric formulation provides a natural setting for investigating fundamental information-theoretic tasks and operational notions of separation \cite{barnaum07,PLAVALA20231,shahanadeh,colbeck23,Wright2021,heinosaari2022random,Barnum_2010,filipov,Heinosaari_2020,barnum2008,massar}. An important motivation for studying such generalized frameworks is to identify which operational features arise from the underlying mathematical structure of a theory and which are specific to quantum mechanics, thereby contributing to the broader investigation of the uniqueness of quantum theory \cite{Masanes_2011,hardy2001,sikora18,patra23,Kolangatt2025,banik19,sikora20,Stevens,pawlowski_ic,Janotta_2011,muller}. In particular, a substantial body of work has explored state distinguishability in GPTs, revealing how its operational properties are governed by the geometry of the state and effect spaces rather than being intrinsically tied to the Hilbert-space structure of quantum mechanics \cite{kimura09,KIMURA2010175,takagi,janotta2014generalized,banik19,Bhattacharya}. These developments naturally motivate the study of antidistinguishability within the GPT framework, where the operational conditions for excluding states can be investigated directly through the geometry of convex state and effect spaces.

In this work, we formulate the notion of antidistinguishability of states within the framework of general probabilistic theories (GPTs). After providing a brief overview of GPTs, we define antidistinguishability in this framework and introduce a stronger notion, termed strong antidistinguishability, which requires the elimination process to exhaust all the states in the set. We further introduce the notion of equal antidistinguishability, which arises when the non-eliminated states have equal overlap with the corresponding POVM element.
We then establish several general results concerning the relations among antidistinguishability, strong antidistinguishability, and distinguishability. A central result of this section determines the maximum cardinality of an equally antidistinguishable set of states.
We subsequently turn to a detailed analysis of antidistinguishability in polygonal theories. After establishing several useful lemmas, we derive a condition under which a set of states in a polygon model is antidistinguishable. As an immediate consequence, we show that every set of pure states in a polygon model is antidistinguishable. We then characterize strongly antidistinguishable and equally antidistinguishable sets in polygon models. As expected, all these results continuously reduce to the corresponding results for the qubit, represented by the disk model, in the limit where the number of vertices of the polygon tends to infinity.
Finally, we analyse a two-party communication task, termed Random Exclusion Codes (REC) \cite{excl}, to quantify the nonclassicality of polygon models. The success merit of REC is connected with antidistinguishability of the different sets of encoding states. We show that every polygon model exhibits nonclassical behavior and, moreover, that several polygon models can achieve a success probability exceeding the quantum optimum. Nevertheless, as the number of vertices increases, the optimal success merit of the polygon model converges to the quantum optimum.

We begin with a brief overview of generalized probabilistic theories, followed by the definitions of distinguishability and antidistinguishability. In Sec.~\ref{gen_res}, we present our general results concerning these notions. We then specialize to polygonal theories in Sec.~\ref{n_gon}, where we investigate their antidistinguishability properties in detail. In Sec.~\ref{rec}, we introduce the Random Exclusion Code task and study its performance in polygonal theories. Finally, we summarize our main findings in the conclusion and discuss open problems and possible directions for future research.

\section{Preliminaries}
\subsection{General Probabilistic Theory}
In general probabilistic theory, an elementary system is denoted by a tuple $(\Omega,\mathcal{E})$.

\textit{ State space--}
The state space $\Omega$ of a system is a convex, compact set embedded in the positive cone $V_+$ of some real vector space $V$.  Convexity ensures that statistical mixtures of valid preparations are also valid preparations, i.e., for any $\omega_1, \omega_2 \in \Omega$, any probabilistic mixture
\[
\omega=p \omega_1 + (1-p)\omega_2 \in \Omega, \quad \text{where } p \in [0,1].
\]
We call $\omega$ is a \emph{pure state} if $\omega=\omega_1=\omega_2$. We call $\omega$ is a \emph{mixed state} if it is not a pure state. 

Pure states are those that cannot be prepared by randomizing over different preparation procedures. Therefore, a pure state $\omega$ corresponds to a deterministic, non-randomized preparation. Mixed states, on the other hand, arise from probabilistic mixtures of preparation procedures and constitute the counterpart to pure states.

The set $\Omega$ is assumed to be topologically closed. Furthermore, finite dimensionality of $V$ guarantees compactness of $\Omega$.

\textit{ Effect space---}
Effects are linear functionals on $\Omega$ that map each state to a probability. The set of all such linear functionals is denoted by $\Omega^* \subset V_+^*$\cite{cheridito2013}. The framework of generalized probabilistic theories may assume that not all mathematically well-defined states and observables are physically implementable. For example, the set of physically allowed effects $\mathcal{E}$ may be a strict subset of $\Omega^*$. A theory in which all elements of $\Omega^*$ are allowed effects is called \emph{dual}.

A $n$-outcome measurement $M$ is specified by a collection of $n$ effects:
\[
M \equiv \{ e_j \}_{j=1}^n \quad \text{such that} \quad \sum_{j=1}^n e_j = u,
\]
where $u$ is the unit effect. $u\in\Omega^*$ is a constant function given as $u(\omega)=1$, for any $\omega\in\Omega$. That also implies that $\omega$ is normalized state. So, all the states residing at state-space are normalized.

Similar to the states, effects can be pure or mixed.
\begin{definition}
    (Face identified by a state)
\end{definition}
The face identified
by a state $\omega\in\Omega$ is the set $F_\omega$ of all states $\sigma\in\Omega$
such that $\omega = p\sigma + (1-p)\tau$, for some nonzero probability $p > 0$ and some state $\tau\in\Omega$.

In other words, $F_\omega$ is the set of all  states that show up in the convex decompositions of $\omega$. Clearly, if $\Phi$ is a
pure state, then one has $F_\Phi=\{\Phi\}$. The opposite situation is that
of completely mixed states.
\begin{definition}
    (Completely mixed state) 
\end{definition}
A state $\Sigma$
is completely mixed if every state $\sigma\in\Omega$ can stay in its
convex decomposition, that is, if $F_{\Sigma}=\Omega$.
\begin{lemma}\label{l1}
    Any effect other than zero effect acting on completely mixed state will produce non-zero probability.
\end{lemma}
\begin{proof}
    Suppose, there exists an effect $\epsilon$ such that $\epsilon(\Sigma)=0$, where $\Sigma$ is a completely mixed state which can be written as $\Sigma=p\sigma+(1-p)\tau$, for all $\sigma\in\Omega$. Our assumption indicates $\epsilon(\sigma)=0$ for all $\sigma$. As we know $\sigma$ can be any state, that means there must exist one effect which produces null probability for all the states of $\Omega$. This is only possible if the effect is a zero effect.
\end{proof}
\textit{Transformation---}
Transformations map states to states (and effects to effects), i.e.,
\[
T : V \to V, \quad \text{with } T(V_+) \subseteq V_+.
\]
They are linear and preserve statistical mixtures. Moreover, they cannot increase the total probability, although they are allowed to decrease it.

\textit{Joint system---} For two subsystems,
$
\mathrm{Sys}(A)\equiv(\Omega_A,E_A)
\qquad\text{and}\qquad
\mathrm{Sys}(B)\equiv(\Omega_B,E_B),
$
a GPT also defines a composite system
$
\mathrm{Sys}(AB)\equiv(\Omega_{AB},E_{AB}).
$
The joint state space $\Omega_{AB}$ is required to be convex. Without imposing further restrictions, the relationship between the state spaces of the individual systems and that of their composite can be highly nontrivial. 
For physically meaningful composite theory, the joint state space $\Omega_{AB}$ can be represented within the positive cone of the tensor-product vector space $V_A\otimes V_B$\cite{hardy2001,hardy2012limited}. Moreover, the admissible composite state space is constrained between two extremal constructions, namely the \emph{minimal} and \emph{maximal tensor products}\cite{wilce1992, PLAVALA20231,barnaum07}.
As our work does not need the joint system formulation, we avoid the further details.

\subsection{Distinguishability}
A set of states $\{\omega_i\}_{i=1}^n$ is perfectly distinguishable, if there exists a $n$-outcome measurement $M$ having effects $\{e_j\}_{j=1}^n$ such that $e_j(\omega_i)=\delta_{ij}$. 

\begin{lemma}
    A completely mixed state can not be perfectly distinguished with respect to any other state.
\end{lemma}
\begin{proof}
    We consider the distinguishability of completely mixed state $\Sigma$ and any other state $\omega$. For perfect discrimination, there exists a two-outcome measurement having effects $e_1$ and $e_2$ such that $e_1(\omega)=0$, $e_2(\omega)=1$, $e_1(\Sigma)=1$ and $e_2(\Sigma)=0$. From lemma \ref{l1}, we know that $e_2(\Sigma)\neq 0$, for any $e_2$. This completes the proof.
\end{proof}

\textit{Axiom of perfect distinguishability \cite{chiribella}---}  If a state is not completely mixed (i.e. if it cannot be obtained as a mixture from any other state), then there exists at
 least one state that can be perfectly distinguished from it maximal distinguishable set.

One direct implication of this axiom is there must exist at least two perfectly distinguishable states (Lemma 20 of \cite{chiribella}).

\section{Antidistinguishability}
\begin{definition}(Antidistinguishability)
     A set of $n$ states $\{\omega_i\}_i$ are antidistinguishable if there exists a measurement such that $e_i(\omega_i)=0$, $\forall j$. The definition can be translated into a linear function as follows:
     \bea
\mathcal{A}[\{\omega_i\}_i]=\max_{\{e_j\}_j} \Bigg\{ \frac1n\sum_{j,i}  p(j\neq i|\omega_i,e_j)\Bigg\}
     \eea
 the above expression becomes,
\bea\label{pA}
\mathcal{A}[\{\omega_i\}_i] 
&=& 1 - \frac{1}{n}\min_{\{e_i\}_i}\left\{\sum_{i} \ e_i(\omega_i) \right\},
\eea 
where $\{e_i\}_i$ are the effects of the optimum measurement. Note that $e_i$ can be $\mathbf{0}$ for some $i$, where $\mathbf{0}$ is the zero effect.
\end{definition}

\begin{definition}
    (Strong Antidistinguishability) A set of $n$ states $\{\omega_i\}_i$ are strongly antidistinguishable if there exists a measurement such that $e_i(\omega_i)=0$, $\forall i$ and $e_i\neq\mathbf{0},\forall i$.
\end{definition}
The same linear expression of \eqref{pA} is also applicable here with the extra condition such that none of the effects is a zero effect. These definitions are already demonstrated in ref. \cite{manna20}.
\begin{definition}
    (Equally antidistinguishable set)
\end{definition}
A set of $n$ states $\{\omega_i\}_i$ are equally antidistinguishable if the measurement $\{e_i\}_i$ such that $e_i(\omega_i)=0$ and $e_i(\omega_j)=\frac{1}{n-1}$, where $j\neq i$.
By definition, equally antidistinguishable set is a strongly antidistinguishable set. For an example, consider equally antidistinguishable set whose cardinality is $3$. For such a set, the matrix elements of $e_i(\omega_j)$ can be written as,
\be
e_i(\omega_j)=
\begin{pmatrix}
    0 & 1/2 & 1/2\\
    1/2 & 0 & 1/2\\
    1/2 & 1/2 & 0
\end{pmatrix}
\ee
\begin{obs}
    From the definition, we can infer 
\bea
&&\text{Equal antidistinguishability}\nonumber\\
&&\implies \text{strong antidistinguishability}\nonumber\\ 
&&\implies \text{antidistinguishability}.
\eea
\end{obs}
\section{General Results}\label{gen_res}
\begin{lemma}
    Let $\{\mathbbm{S}_i\}_i$ are the sets of states. The sufficient condition for $\bigcup_i\mathbbm{S}_i$ to be antidistinguishable if one of the set from $\mathbbm{S}_i$ is antidistinguishable but this is not a necessary condition.
\end{lemma}
\begin{proof}
   Suppose, $\mathbbm{S}_{i=i*}$ is an antidistinguishable set. The antidistinguishing measurement has effects $\{g_j\}_{j=1}^r$, where $r$ is the cardinality of $\mathbbm{S}_{i=i*}$. The set $\bigcup_i\mathbbm{S}_i$ is also antidistinguishable with the measurement having effects $\{\{g_j\}_{j=1}^r,\{\mathbf{0}\}_{j=r+1}^{r'}\}$, where $r'$ is the cardinality of $\bigcup_i\mathbbm{S}_i$. The sufficiency is proved.

   Now consider two not antidistinguishable sets $\mathbbm{S}_1=\{\rho_1,\rho_2,\rho_3\}$ and $\mathbbm{S}_2=\{\rho_4,\rho_5,\rho_6\}$. In this construction, take $\rho_1$ and $\rho_4$ are distinguishable as we know there exists at least two perfectly distinguishable states (Axiom of perfect distinguishability). Therefore, there exists a measurement $\{f_1,f_2\}$ such that $f_1(\rho_1)=f_2(\rho_4)=0$ and $f_1(\rho_4)=f_2(\rho_1)=1$. Now, $\mathbbm{S}_1\bigcup\mathbbm{S}_2=\{\rho_1,\rho_2,\rho_3,\rho_4,\rho_5,\rho_6\}$. This set is antidistinguishable with the measurement having effects $\{f_1,\mathbf{0}, \mathbf{0}, f_2, \mathbf{0}, \mathbf{0}\}$. This completes the proof.
\end{proof}

The statement of this lemma can be presented in other way around. Thus the next corollary follows. 
\begin{coro}\label{cor1}
    If a subset of a set of states is antidistinguishable, the set is antidistinguishable. This implies the fact that any antidistinguishable set can be extended for arbitrary number of states. The maximum cardinality of an antidistinguishable set can be infinite.
\end{coro}
\begin{lemma}
    Suppose, the states $\eta_1,\eta_2$ and $\eta_3$ belongs to the face of $\rho_1,\rho_2$ and $\rho_2$ respectively. If the set $\{\rho_1,\rho_2,\rho_3\}$ are antidistinguishable, the set $\{\eta_1,\eta_2,\eta_3\}$ is also antidistinguishable.
\end{lemma}
\begin{proof}
    The measurement with effects $\{g_1,g_2,g_3\}$ is the antidistinguishing measurement for $\{\rho_1,\rho_2,\rho_3\}$. From the definition, we can write $g_i(\rho_i)=0$, for all $i$. We know $\rho_i=p_i\eta_i+(1-p_i)\tau_i$. It can easily seen that $g_i(\eta_i)=0$ as $g_i(\rho_i)=0$. So $\{\eta_1,\eta_2,\eta_3\}$ is antidistinguishable with the same measurement.
\end{proof}
This result can be trivially generalized for any number of states.
\begin{lemma}
A distinguishable set is a strongly antidistinguishable set.    
\end{lemma}
\begin{proof}
    A set consists of $\{\rho_i\}_{i=1}^n$ is distinguishable if there exists a measurement with effects $\{g_i\}_{i=1}^n$ such that $g_i(\rho_j)=\delta_{ij}$. To antidistinguish this set, the same measurement is executed. If the outcome corresponding to $g_i$ clicks, the eliminated state is $\rho_{j\neq i}$. From the definition of distinguishability, it is easy to see $g_i\neq 0, \forall i$. Therefore, the set is strongly antidistinguishable.
\end{proof}
\begin{lemma}
     Let $\{\mathbbm{S}_i\}_i$ are the sets of strongly antidistinguishable states. $\bigcup_i\mathbbm{S}_i$ is also strongly antidistinguishable.
\end{lemma}
\begin{proof}
    Suppose, the strongly antidistinguishing measurement for the set $\mathbbm{S}_i$ is $M^i=\{g^i_1,\cdots,g^i_{r_i}\}$. Consequently, $g_i> 0$ and $\sum_{i=1}^{r_i} g_{i}=u$. For the set, $\bigcup_i\mathbbm{S}_i$, the strongly antidistinguishing measurement would be $\Tilde{M}=\frac{1}{|i|}\bigcup_i\{g^i_1,\cdots,g^i_{r_i}\}$. It can be easily checked that $\frac{1}{|i|}g^i_x> 0$ and $\frac{1}{|i|}\sum_i g^i_{r_i}=\frac{1}{|i|} |i|u=u$. $|i|$ is the number of sets.
    \end{proof}
\begin{lemma}
A set consists of completely mixed state can be antidistinguishable but never be strongly antidistinguishable.   
\end{lemma}
\begin{proof}
   From Lemma \ref{l1}, no effect other than the zero effect can yield zero probability when acting on the completely mixed state. By the definition of strong antidistinguishability, there must exist a set of effects $\{g_i\}_i$ such that $g_i(\rho_i)=0$ for all $i$. This condition cannot be satisfied for the state $\rho_{i^*}$ that is completely mixed. Therefore, any set containing the completely mixed state cannot be strongly antidistinguishable.

Nevertheless, such a set may still be antidistinguishable. In particular, this is possible when a subset of the states, excluding the completely mixed state, is antidistinguishable (From Corollary \ref{cor1}).
\end{proof}

\begin{thm}\label{th1}
Let
$
\mathcal{S}=\{\rho_1,\ldots,\rho_k\}
$
be a set of $k$ distinct states. Suppose that $\mathcal{S}$ is equally antidistinguishable by a measurement
$\{g_i\}_{i=1}^k$
Then
$k\leq d+1,$ where $d$ is the affine dimension of the state space.
\end{thm}

\begin{proof}
Since every GPT effect is an affine functional on the state space, for each $i$ the set
$
H_i=
\left\{
\rho:
g_i(\rho)=\frac{1}{k-1}
\right\}
$
is an affine hyperplane, provided $e_i$ is nonconstant. By
definition of equally antidistinguishable set, all states except $\rho_i$ belong to this hyperplane:
$
\rho_j\in H_i,\qquad j\neq i.
$
Moreover, $\rho_i\notin H_i$, since
$
g_i(\rho_i)=0\neq\frac{1}{k-1}.
$
We now show that the states ${\rho_1,\ldots,\rho_k}$ must be affinely independent. Suppose, to the contrary, that they are affinely dependent. Then there exist real numbers
$c_1,\ldots,c_k$, not all zero, such that
\begin{equation}
\sum_{j=1}^k c_j=0,
\qquad
\sum_{j=1}^k c_j\rho_j=0.
\label{eq:affine_dependence}
\end{equation}
For any fixed $i$, all $\rho_j$ with $j\neq i$ lie in the affine hyperplane $H_i$. Hence, applying the affine functional $g_i$ to Eq.~\eqref{eq:affine_dependence} gives
$
\sum_{j=1}^k c_j g_i(\rho_j)=0.
$
Using the definition, we obtain
$
c_i\,0+
\frac{1}{k-1}\sum_{j\neq i}c_j=0.
$
Since $\sum_j c_j=0$, this implies
$
-\frac{c_i}{k-1}=0,
$
and therefore
$
c_i=0.
$
Since $i$ was arbitrary, we obtain
$
c_1=c_2=\cdots=c_k=0,
$
contradicting the assumed affine dependence.

Thus the states $\rho_1,\ldots,\rho_k$ are affinely independent. A $d$-dimensional affine space can contain at most $d+1$ affinely independent points. Consequently,
$
k\leq d+1.
$
This completes the proof.
\end{proof}

\section{Antidistinguishability in Polygonal Model}\label{n_gon}
Polygonal model is described as,
$\mathrm{P}(n) \equiv (\Omega(n), \mathcal{E}(n))$, where $n$ is the number of pure states in the corresponding theory. The state
spaces \(\Omega(n)\) for elementary systems are regular polygons with $n$ vertices. The states and effects are represented by vectors in $\mathbbm{R}^3$.

 For any $n$, \(\Omega(n)\) is the convex hull of \(n\) pure states
\(\{\omega_i\}_{i=0}^{n-1}\) with
\bea\label{omega}
\omega_i :=
\begin{pmatrix}
r_n \cos\left(\frac{2\pi i}{n}\right) \\
r_n \sin\left(\frac{2\pi i}{n}\right) \\
1
\end{pmatrix},
\eea
where
$r_n := \sqrt{\sec\left(\frac{\pi}{n}\right)}.$

The unit effect is described by
\bea
u := 
\begin{pmatrix}
0 \\
0 \\
1
\end{pmatrix}.
\eea

The set \(\mathcal{E}(n)\) is the convex hull of the zero effect, the unit effect, and the extremal effects
\(\{e_i,\bar e_i=u-e_i\}_{i=0}^{n-1}\). For even $n$, 
\bea
e_i :=
\frac{1}{2}
\begin{pmatrix}
r_n \cos\left(\frac{(2i+1)\pi}{n}\right)\\
r_n \sin\left(\frac{(2i+1)\pi}{n}\right)\\
1   
\end{pmatrix}.
\eea
For odd \(n\),
\bea
e_i :=
\frac{1}{1+r^2_n}
\begin{pmatrix}
r_n \cos\left(\frac{2\pi i}{n}\right)\\
r_n \sin\left(\frac{2\pi i}{n}\right)\\
1   
\end{pmatrix}.
\eea
 It is easy to see that, for even n-gon, $\Bar{e}_i=e_{i+\frac{n}{2}(\mod n)}$. The similar relationship in odd n-gon stands out as $\Bar{e}_i=\frac{r_n^2}{2}\Big(e_{i+\frac{n-1}{2}(\mod n)}+e_{i+\frac{n+1}{2}(\mod n)}\Big)$. 

In the limit $n\rightarrow\infty$, the vertices of the regular $n$-gon become dense on a circle, while their convex hull approaches the corresponding filled disk. Thus, the limiting state space $\Omega(n)$ is the Bloch disk, which can be identified with the set of equatorial qubit states in quantum theory.
 
 \subsubsection{Results of even n-gon}
\begin{lemma}\label{5}
    One pure effect can eliminate two pure states.
\end{lemma}
\begin{proof}
    For an even $n$-gon, $e_j(\omega_i)=0$ implies
\bea
&&\frac12\Big(r_n^2\cos\left(\frac{2\pi i}{n}\right)\cos\left(\frac{(2j+1)\pi}{n}\right)\nonumber\\
&&
+r_n^2\sin\left(\frac{2\pi i}{n}\right)\sin\left(\frac{(2j+1)\pi}{n}\right)+1\Big)=0\nonumber\\
&\implies&
\cos\left(\frac{\pi}{n}(2i-2j-1)\right)+\cos\frac{\pi}{n}=0.
\eea
Since
$
\cos x=-\cos\frac{\pi}{n}
=\cos\left(\pi\pm\frac{\pi}{n}\right),
$
we obtain
\bea
\frac{\pi}{n}(2i-2j-1)
=\pi\pm\frac{\pi}{n}
\pmod{2\pi},
\eea
which is equivalent to
\bea
j=\left(i+\frac{n}{2}\right)\pmod n
\qquad\text{or}\qquad\nonumber\\
j=\left(i+\frac{n}{2}-1\right)\pmod n.\nonumber\\
\eea
Hence, every pure effect vanishes on exactly two pure states.
\end{proof}

\begin{coro}\label{even_c2}
    Two eliminated states are $\frac{2\pi}{n}$ radian apart.
\end{coro}
\begin{proof}
    Angular separation of two states is $(\frac{2\pi(j+n/2+1)}{n}-\frac{2\pi(j+n/2)}{n})=\frac{2\pi}{n}$. Similarly, angular separation of two vanishing effects is also $2\pi/n$.
\end{proof}
As $n$ grows larger, this two eliminated states come closer. At $n\rightarrow\infty$, they become one, which is precisely the same case as for qubit states.

\begin{thm}
   The pure states $\{\omega_i\}_i$ are antidistinguishable if and only if there exist coefficients $c_i\ge 0$ such that
\be\label{cond}
\sum_i c_i Con\{\omega_{(i+n/2-q+1/2)\pmod n}\}_q=2(0,0,1)^T,
\ee
where $q\in\{0,1\}$ and $Con\{f_k\}_k$ denotes the convex combination of vectors $\{f_k\}_k$.
\end{thm}
\begin{proof}
For a set of states $\{\omega_i\}_{i}$ is antidistinguishable if and only if there exist positive coefficients $\{c_j\}_{j}$ such that
\begin{equation}
\sum_{j} c_j e_j=u,
\end{equation}
where $e_j(\omega_i)=0$ and $u=(0,0,1)^T$ is the unit effect. From previous lemma, we know the relation between $i$ and $j$. Using this, the above equation becomes,
\begin{equation}\label{eq:anti-meas}
\sum_i c_{i} Con\{e_{(i+n/2-q) \pmod n}\}_q=u, 
\end{equation}
where $q\in\{0,1\}$.
Using the form of state and effect, 
$e_i=\dfrac{1}{2}\,\omega_{i+\frac12}$,
Eq.~\eqref{eq:anti-meas} becomes
\begin{equation}
\frac12\sum_i c_iCon\{\omega_{(i+n/2-q+1/2) \pmod n}\}_q=u,
\end{equation}
By simplification, we get \eqref{cond}.
For strong antidistinguishability, $c_i\neq 0, \forall i$. It is also easy to see that $\sum_i c_i=2$.
\end{proof}
This theorem invites an immediate and important corollary.
\begin{coro}
    Set of all pure states are strongly antidistinguishable.
\end{coro}
\begin{proof}
In Eq. \eqref{cond}, fixing $r=0$, we find that 
\begin{equation}
\sum_i c_i\omega_{(i+n/2+1/2) \pmod n}=2u.
\end{equation}
Therefore, we have to find $c_i$ such that the above condition is satisfied.
We choose $c_i=2/n$, for all $i$. Then one can check that, from the above equation and \eqref{omega}, 
\bea
\frac{r_n}{n}\sum_i \cos(\pi+\frac{\pi}{n}+\frac{2\pi\mathbbm{i}}{n})&=& 0;\nonumber\\
\frac{r_n}{n}\sum_i \sin(\pi+\frac{\pi}{n}+\frac{2\pi\mathbbm{i}}{n})&=& 0;\nonumber\\
\frac1n\sum_i 1=1;
\eea
Above three equations are proved to be true. Thus our corollary is established. 
\end{proof}

Now we chracterize the strongly antidistinguishable set in even $n$-gon.

\begin{thm} 
Let $k\geq 3$ be a divisor of $n$. For any $a\in\{0,1,\ldots,n-1\}$, consider the set of $k$ equally spaced pure states
\begin{equation}
\mathcal{S}_{k,a}
:=
\left\{
\omega_{a+\ell n/k}:
\ell=0,1,\ldots,k-1
\right\},
\end{equation}
where all indices are understood modulo $n$. Then $\mathcal{S}_{k,a}$ is strongly antidistinguishable.
\end{thm}

\begin{proof}
We construct explicitly a measurement
$\{f_\ell\}_{\ell=0}^{k-1}$ satisfying
\begin{equation}
f_\ell(\omega_{a+\ell n/k})=0,
\qquad
f_\ell\neq 0,
\qquad
\sum_{\ell=0}^{k-1}f_\ell=u,
\end{equation}
where,
\begin{equation}
f_\ell
:=
\frac{2}{k}
\overline e_{a+\ell n/k},
\qquad
\ell=0,1,\ldots,k-1.
\end{equation}
Since $k\geq3$, we have $2/k\leq1$, and hence $f_\ell$ is a valid nonzero effect. Moreover,
\begin{equation}
f_\ell
\left(
\omega_{a+\ell n/k}
\right)
=
\frac{2}{k}
\overline e_{a+\ell n/k}
\left(
\omega_{a+\ell n/k}
\right)
=0.
\end{equation}

It remains to show that the effects form a measurement. Using
$\overline e_j=u-e_j$, we have
\begin{align}
\sum_{\ell=0}^{k-1}\overline e_{a+\ell n/k}
&=
\sum_{\ell=0}^{k-1}
\frac{1}{2}
\begin{pmatrix}
r_n\cos\left(
\frac{(2a+2\ell n/k+n)\pi}{n}
\right)\\
r_n\sin\left(
\frac{(2a+2\ell n/k+n)\pi}{n}
\right)\\
1
\end{pmatrix}\nonumber\\
&=
\frac{k}{2}
\begin{pmatrix}
0\\
0\\
1
\end{pmatrix}
=
\frac{k}{2}u.
\end{align}
Here we have used
$
\sum_{\ell=0}^{k-1}
e^{2\pi i\ell/k}=0,
\qquad k\geq2.
$
Consequently,
\begin{equation}
\sum_{\ell=0}^{k-1}f_\ell
=
\frac{2}{k}
\sum_{\ell=0}^{k-1}\overline e_{a+\ell n/k}
=u.
\end{equation}
Thus, $\mathcal{S}_{k,a}$ is strongly antidistinguishable for even $n$.

\end{proof}

\begin{example}[Trine]
For $k=3$ and $a=0$, with $n=6m$, where $m\in\mathbb{N}$, the theorem gives the trine
\[
\mathcal{S}_{3,0}
=
\{\omega_0,\omega_{n/3},\omega_{2n/3}\}.
\]
These three states are equally spaced by an angle $2\pi/3$. The
measurement prescribed by the theorem is
\[
\left\{
\frac23\overline e_0,\,
\frac23\overline e_{n/3},\,
\frac23\overline e_{2n/3}
\right\}
=
\left\{
\frac23 e_{n/2},\,
\frac23 e_{5n/6},\,
\frac23 e_{7n/6}
\right\},
\]
where we used $\overline e_j=e_{j+n/2}$.
Hence the trine is strongly antidistinguishable.
\end{example}
One can check
\bea
&&\frac23 e_{n/2}(\omega_{n/3})=\frac13 (1+\sec(\pi/n)\cos(\pi/3+\pi/n)).\nonumber\\
\eea
Similarly,
$\frac23 e_{n/2}(\omega_{2n/3})= \frac13 (1+\sec(\pi/n)\cos(\pi/3-\pi/n)),
\frac23 e_{5n/6}(\omega_{0})=\frac13 (1+\sec(\pi/n)\cos(5\pi/3+\pi/n)),
\frac23 e_{5n/6}(\omega_{2n/3})=\frac13 (1+\sec(\pi/n)\cos(\pi/3+\pi/n)),
\frac23 e_{7n/6}(\omega_{0})=\frac13 (1+\sec(\pi/n)\cos(7\pi/3+\pi/n))$ and $\frac23 e_{7n/6}(\omega_{n/3})=\frac13 (1+\sec(\pi/n)\cos(5\pi/3+\pi/n));$

For finite \(m\), these probabilities are generally not equal. Thus, although the set is strongly antidistinguishable, it does not appear to be equally antidistinguishable with the measurement considered above. In the limit \(n\rightarrow\infty\), however, all these probabilities approach \(1/2\), and the set becomes equally antidistinguishable, as expected for the qubit case. This naturally raises the question of whether such a set can be found in an even \(n\)-gon. Interestingly, the above set is indeed equally antidistinguishable for finite \(m\); the key is to choose a suitable measurement that makes the non-vanishing probabilities equal, i.e., $1/2$.

\begin{thm}
$
\mathcal{S}_{3,0}
=
\{\omega_0,\omega_{n/3},\omega_{2n/3}\}
$
 is equally antidistinguishable. 
\end{thm}

\begin{proof}
Consider the following three effects:
\begin{align}
f_0&=
\frac{1}{3}
\left(
-\frac{1}{r_n},
0,
1
\right),\nonumber\\
f_{n/3}&=
\frac{1}{3}
\left(
-\frac{\cos(2\pi/3)}{r_n},
-\frac{\sin(2\pi/3)}{r_n},
1
\right),\nonumber\\
f_{2n/3}&=
\frac{1}{3}
\left(
-\frac{\cos(4\pi/3)}{r_n},
-\frac{\sin(4\pi/3)}{r_n},
1
\right).
\end{align}
We have
$
f_0+f_{n/3}+f_{2n/3}=(0,0,1)=u,
$
so that \(\{f_0,f_{n/3},f_{2n/3}\}\) forms a measurement.

Next, evaluating \(f_i\) on an arbitrary state \(\omega_j\), we obtain
\begin{align}
f_i(\omega_j)
&=
\frac{1}{3}
\left[
-\cos\theta_i\cos\theta_j
-\sin\theta_i\sin\theta_j+1
\right]\nonumber\\
&=
\frac{1}{3}
\left[
1-\cos(\theta_j-\theta_i)
\right].
\end{align}
Thus $\{f_i\}_i$ are valid effects as $0\leq \frac13\left[
1-\cos(\theta_j-\theta_i)
\right]\leq \frac23<1$.
For \(j=i\), this gives
$
f_i(\omega_i)
=
\frac{1}{3}[1-\cos(0)]
=0.
$
On the other hand, for \(j\neq i\), the angular separation between the states is \(2\pi/3\), and hence
$
\cos(\theta_j-\theta_i)
=
\cos\frac{2\pi}{3}
=-\frac{1}{2}.
$
Therefore,
$
f_i(\omega_j)
=
\frac{1}{3}
\left(1+\frac{1}{2}\right)
=\frac{1}{2},
\qquad i\neq j.
$
Thus, the conditional probabilities generated by the measurement are
$$
\begin{array}{c|ccc}
 & \omega_0 & \omega_{n/3} & \omega_{2n/3}\\
\hline
f_0 & 0 & \frac12 & \frac12\\
f_{n/3} & \frac12 & 0 & \frac12\\
f_{2n/3} & \frac12 & \frac12 & 0
\end{array}.
$$

Hence each outcome \(i\) occurs with zero probability on the corresponding state \(\omega_i\), while it occurs with strictly same probability on the other two states. Therefore,

$
f_i(\omega_i)=0,\qquad
f_i(\omega_j)=1/2\quad (i\neq j),
$
which proves that \(\mathbbm{S}\) is equally antidistinguishable. 
\end{proof}

\subsubsection{Results of odd $n$-gon}
\begin{lemma}\label{10}
    One pure effect can eliminate two pure states.
\end{lemma}
\begin{proof}
    For an odd $n$-gon, $e_j(\omega_i)=0$ implies
\bea
&&\frac{1}{1+r_n^2}\Big(r_n^2\cos\left(\frac{2\pi i}{n}\right)\cos\left(\frac{2\pi j}{n}\right)\nonumber\\
&&\qquad\qquad
+r_n^2\sin\left(\frac{2\pi i}{n}\right)\sin\left(\frac{2\pi j}{n}\right)+1\Big)=0\nonumber\\
&\implies&
\cos\left(\frac{2\pi}{n}(i-j)\right)+\cos\frac{\pi}{n}=0.
\eea
Since
$
\cos x=-\cos\frac{\pi}{n}
=\cos\left(\pi\pm\frac{\pi}{n}\right),
$
we obtain
\be
\frac{2\pi}{n}(i-j)
=\pi\pm\frac{\pi}{n}
\pmod{2\pi},
\ee
which is equivalent to
\be
j=\left(i-\frac{n}{2}\pm\frac12\right)\pmod n.
\ee

Since $u-e_j$ is also a pure effect in an odd $n$-gon, we must also consider the possibility that
$
(u-e_j)(\omega_i)=0.
$
This condition is equivalent to
$
\cos\left(\frac{2\pi}{n}(i-j)\right)=0.
$
However, this equation has no solution when $n$ is odd, since it would require
$
i-j\equiv \frac{n}{4}\ \text{or}\ \frac{3n}{4}\pmod n,
$
which is impossible for odd $n$. Therefore, the only vanishing pure effects are those satisfying
\[
j=\left(i-\frac{n}{2}\pm\frac12\right)\pmod n.
\]
\end{proof}
\begin{coro}\label{odd_c3}
    Two eliminated states are $\frac{2\pi}{n}$ radian apart.
\end{coro}
\begin{proof}
    Angular separation of two states is $(\frac{2\pi(j+n/2+1/2)}{n}-\frac{2\pi(j+n/2-1/2)}{n})=\frac{2\pi}{n}$.
\end{proof}
As $n$ grows larger, this two eliminated states come closer. At $n\rightarrow\infty$, they become one, which is precisely the same case as for qubit states.
\begin{lemma}
   The pure states $\{\omega_i\}_{i\in(0,..,n-1)}$ are antidistinguishable if and only if there exist coefficients $c_i\ge 0$ such that
\be\label{cond_odd}
\sum_i c_i Con\{\omega_{(i-n/2\pm q)\pmod n}\}_q=(1+r_n^2)(0,0,1)^T,
\ee
where $q\in\{1/2,-1/2\}$ and $Con\{f_k\}_k$ denotes the convex combination of vectors $\{f_k\}_k$.
\end{lemma}
\begin{proof}
For a set of states $\{\omega_i\}_{i}$ is antidistinguishable if and only if there exist positive coefficients $\{c_j\}_{j}$ such that
\begin{equation}
\sum_{j} c_j e_j=u,
\end{equation}
where $e_j(\omega_i)=0$ and $u=(0,0,1)^T$ is the unit effect. From previous lemma, we know the relation between $i$ and $j$. Using this, the above equation becomes,
\begin{equation}\label{anti-meas}
\sum_i c_{i} Con\{e_{(i-n/2\pm q) \pmod n}\}_q=u, 
\end{equation}
where $q\in\{1/2,-1/2\}$.
Using the form of state and effect, 
$e_i=\dfrac{1}{1+r_n^2}\,\omega_{i}$.
Eq.~\eqref{anti-meas} becomes
\begin{equation}
\dfrac{1}{1+r_n^2}\sum_i c_iCon\{\omega_{(i-n/2\pm q) \pmod n}\}_q=u,
\end{equation}
By simplification, we get \eqref{cond_odd}.
For strong antidistinguishability, $c_i\neq 0, \forall i$.
\end{proof}
\begin{coro}
    Set of all pure states are strongly antidistinguishable.
\end{coro}
\begin{proof}
In Eq. \eqref{cond_odd}, fixing $r=1/2$, we find that 
\begin{equation}
\sum_i c_i\omega_{(i-(n-1)/2) \pmod n}=(1+r_n^2)u.
\end{equation}
Therefore, we have to find $c_i$ such that the above condition is satisfied.
We choose $c_i=(1+r_n^2)/n$, for all $i$. Then one can check that, from the above equation and \eqref{omega}, 
\bea
\frac{r_n}{n}\sum_i \cos(\pi-\frac{\pi}{n}-\frac{2\pi\mathbbm{i}}{n})&=& 0;\nonumber\\
\frac{r_n}{n}\sum_i \sin(\pi-\frac{\pi}{n}-\frac{2\pi\mathbbm{i}}{n})&=& 0;\nonumber\\
\frac1n\sum_i 1=1;
\eea
Above three equations are proved to be true.
\end{proof}
Similar to the case of even $n$-gon, we characterize the set of strongly antidistinguishable sets.
\begin{thm} 
Let $k\geq 3$ be a divisor of $n$. For any $a\in\{0,1,\ldots,n-1\}$, consider the set of $k$ equally spaced pure states
\begin{equation}
\Bar{\mathcal{S}}_{k,a}
:=
\left\{
\omega_{a+\ell n/k}:
\ell=0,1,\ldots,k-1
\right\},
\end{equation}
where all indices are understood modulo $n$. Then $\mathcal{S}_{k,a}$ is strongly antidistinguishable.
\end{thm}

\begin{proof}
We construct explicitly a measurement
$\{f_\ell\}_{\ell=0}^{k-1}$ 
where
\begin{equation}
f_\ell
:=
\frac{1+r_n^2}{k}e_{j_\ell},
\end{equation}
and
$
j_\ell
:=
a+\frac{\ell n}{k}+\frac{n-1}{2}.
$
Since $n$ is odd, $(n-1)/2$ is an integer. Hence
$
j_\ell-\frac{n-1}{2}
=
a+\frac{\ell n}{k},
$
and therefore
\begin{equation}
e_{j_\ell}
\left(
\omega_{a+\ell n/k}
\right)
=0.
\end{equation}
Since $n\geq3$ and $k\geq3$,
$
\frac{1+r_n^2}{k}
=
\frac{1+\sec(\pi/n)}{k}
\leq1,
$
with equality only for $n=k=3$. Since the effect space contains $0$ and is convex, any scalar multiple $\lambda e$ with $0\leq\lambda\leq1$ is an allowed effect. Thus $f_\ell$ is a valid nonzero effect.

Furthermore,
\begin{equation}
f_\ell
\left(
\omega_{a+\ell n/k}
\right)
=
\frac{1+r_n^2}{k}
e_{j_\ell}
\left(
\omega_{a+\ell n/k}
\right)
=0.
\end{equation}

Finally, using the explicit form of $e_{j_\ell}$, we obtain
\begin{align}
\sum_{\ell=0}^{k-1}e_{j_\ell}
&=
\frac{1}{1+r_n^2}
\sum_{\ell=0}^{k-1}
\begin{pmatrix}
r_n\cos\left(
\frac{2\pi}{n}
\left(
a+\frac{\ell n}{k}+\frac{n-1}{2}
\right)
\right)\\
r_n\sin\left(
\frac{2\pi}{n}
\left(
a+\frac{\ell n}{k}+\frac{n-1}{2}
\right)
\right)\\
1
\end{pmatrix}.
\end{align}
The first two components vanish because
$
\sum_{\ell=0}^{k-1}
e^{2\pi i\ell/k}=0,
\qquad k\geq3.
$
Hence
\begin{equation}
\sum_{\ell=0}^{k-1}e_{j_\ell}
=
\frac{k}{1+r_n^2}u.
\end{equation}
Therefore,
\begin{equation}
\sum_{\ell=0}^{k-1}f_\ell
=
\frac{1+r_n^2}{k}
\sum_{\ell=0}^{k-1}e_{j_\ell}
=u.
\end{equation}
Hence $\Bar{\mathcal{S}}_{k,a}$ is strongly antidistinguishable for odd $n$.
\end{proof}
\begin{example}[Trine]
 For $k=3$ and $a=0$, with $n=6m+3$, , where $m\in\mathbb{N}$, the theorem gives the trine
\[
\Bar{\mathcal{S}}_{3,0}
=
\{\omega_0,\omega_{n/3},\omega_{2n/3}\}.
\]
These three states are equally spaced by an angle $2\pi/3$. The
measurement prescribed by the theorem is
\[
\left\{
\frac{1+r_n^2}{3} e_{(n-1)/2},\,
\frac{1+r_n^2}{3} e_{(5n-3)/6},\,
\frac{1+r_n^2}{3} e_{(7n-3)/6}
\right\}
\]
\end{example}
Similar to the even $n$-gon, one can check this set is not equally antidistinguishable for finite $n$ with this measurement. Although, $n\rightarrow\infty$, the set reduces to the qubit trine. Next, we prove that $\Bar{\mathcal{S}}_{3,0}$ is basically an equally antidistinguishable set with a suitable chosen measurement.
\begin{thm}
$
\Bar{\mathcal{S}}_{3,0}
=
\{\omega_0,\omega_{n/3},\omega_{2n/3}\}
$
 is equally antidistinguishable. 
\end{thm}

\begin{proof}
The proof is similar to the case of even $n$-gon.
We use the same three effects:
\begin{align}
f_0&=
\frac{1}{3}
\left(
-\frac{1}{r_n},
0,
1
\right),\nonumber\\
f_{n/3}&=
\frac{1}{3}
\left(
-\frac{\cos(2\pi/3)}{r_n},
-\frac{\sin(2\pi/3)}{r_n},
1
\right),\nonumber\\
f_{2n/3}&=
\frac{1}{3}
\left(
-\frac{\cos(4\pi/3)}{r_n},
-\frac{\sin(4\pi/3)}{r_n},
1
\right).
\end{align}
we obtain
$
f_i(\omega_i)=0,
f_i(\omega_j)=\frac12 (i\neq j).
$
Thus, each effect excludes one state with certainty and occurs with equal probability on the remaining two states. 
\end{proof}




\section{Random Exclusion Codes}\label{rec}
Random Exclusion Codes (RECs) are a variant of the well-known communication task of Random Access Codes (RACs) \cite{excl}. In this task, Alice is given a string
$x = x_1x_2$ of length $2$, where each symbol $x_y$ is chosen independently and uniformly from the alphabet $\{1,2,3\}$. Thus, $x$ is chosen uniformly at random from the set $\{1,2,3\}^2$.
Alice encodes the string and communicates the resulting message to Bob. After receiving the message, Bob is given an index $y \in \{1,2\}$, chosen uniformly at random. His task is to correctly guess $z\neq x_y$, i.e., actually antidistinguish $x_y$ for a given $y$.
 The success metric is determined by the average success probability, defined through 
\be \label{Snd}
\mathcal{S}_{REC} = \frac{1}{18} \sum_{x,y} p(z\neq x_y|x,y) .
\ee
We analyze the success merit which can be achieved in different polygonal model. In this demonstration, we restrict ourselves in to two-outcome measurements implemented by Bob. We do not have proof of optimality of the outcome number of Bob's measurement. Therefore, our merit works as a lower bound of the actual merit can be obtained in the particular theory. Let us take the encoding states of Alice are $\{\omega_{ab}\}_{a,b=0}^2$ and Bob executes two dichotomic measurements such as $\{e_0,\Bar{e_0}=u-e_0\}$ and $\{e_k,\Bar{e_k}=u-e_k\}$. Bob's first measurement can be taken as the fixed measurement and the optimization can be implemented on Bob's second measurement and all the encoding states of Alice. Success merit takes the form as,
\bea
&&\mathcal{S}_{REC}(n)\nonumber\\
&=& \frac23 +\frac{1}{18}[(e_0+e_k)(\omega_{11}-\omega_{00})+(-e_0+e_k)(\omega_{01}-\omega_{10})]\nonumber\\
&&+\frac{1}{18}[e_0(\omega_{12}-\omega_{02})+e_k(\omega_{21}-\omega_{20})]
\eea
Note that, the second term, i.e., $[(e_0+e_k)(\omega_{11}-\omega_{00})+(-e_0+e_k)(\omega_{01}-\omega_{10})]$ occurs in the merit of Random Access codes \cite{heinosaari2022random}. using the success merit ($\mathcal{S}_{RAC}(n)$) of Random Access Codes, we can write the above equation as,
\bea\label{S_rec}
&&\mathcal{S}_{REC}(n)\nonumber\\
&=& \frac23 +\frac{8}{18}[\mathcal{S}_{RAC}(n)-\frac12]\nonumber\\
&&+\frac{1}{18}[e_0(\omega_{12}-\omega_{02})+e_k(\omega_{21}-\omega_{20})].\nonumber\\
\eea
The ref. \cite{heinosaari2022random} proved the tight bound of $\mathcal{S}_{RAC}(n)$ for every $n\geq 4$. In our case, the highest value of $[e_0(\omega_{12}-\omega_{02})+e_k(\omega_{21}-\omega_{20})]$ is $2$. This is achieved when $e_0(\omega_{12}) =e_k(\omega_{21})=1$ and $e_0(\omega_{02}) =e_k(\omega_{20}) =0$. Now we will prove it is always achievable in polygonal model.

\begin{lemma}
    There exists two pure states $\omega_m$ and $\omega_n$ such that $e_k(\omega_m)=1$ and  $e_k(\omega_n)=0$, where $e_k$ is any particular effect.
\end{lemma}
\begin{proof}
    The second condition comes directly from the Lemma \ref{5} and Lemma \ref{10} for even $n$-gon and odd $n$-gon respectively. We have to prove the first condition. For even $n$-gon,
    \bea
&&\frac12\Big(r_n^2\cos\left(\frac{2\pi m}{n}\right)\cos\left(\frac{(2k+1)\pi}{n}\right)\nonumber\\
&&
+r_n^2\sin\left(\frac{2\pi m}{n}\right)\sin\left(\frac{(2k+1)\pi}{n}\right)+1\Big)=1\nonumber\\
&\implies&
\cos\left(\frac{\pi}{n}(2m-2k-1)\right)-\cos\frac{\pi}{n}=0,
\eea
which gives $m=k$ or $m= (k+1) \mod n$. Similarly, for odd $n$-gon, we find,
\bea
&&\frac{1}{1+r_n^2}\Big(r_n^2\cos\left(\frac{2\pi m}{n}\right)\cos\left(\frac{2\pi k}{n}\right)\nonumber\\
&&\qquad\qquad
+r_n^2\sin\left(\frac{2\pi m}{n}\right)\sin\left(\frac{2\pi k}{n}\right)+1\Big)=1\nonumber\\
&\implies&
\cos\left(\frac{2\pi}{n}(m-k)\right)=1.
\eea
Consequently, $m=k$.
\end{proof}
This lemma suffices that 
\bea\label{Rec_n}
\mathcal{S}_{REC}(n)
&=& \frac23 +\frac{8}{18}[\mathcal{S}_{RAC}(n)-\frac12]+\frac{2}{18}\nonumber\\
&=&\frac59+\frac49\mathcal{S}_{RAC}(n).
\eea
Now we calculate the optimum merit of REC for polygonal theory.
\begin{thm}
    For even $n$-gon, the optimum success merit of REC reads as,
  \begin{align}\label{even_rec}
&\mathcal{S}^{\max}_{\mathrm{REC}}(n)
= \\[-2pt]
&\begin{cases}
\displaystyle \frac{7+\sqrt{2}r_n^2}{9},
& n=4l,\quad l\in\mathbb{N}\ \text{odd},\\[6pt]
\displaystyle \frac{7+\sqrt{2}}{9},
& n=4l,\quad l\in\mathbb{N}\ \text{even},\\[6pt]
\displaystyle \frac{7+r_n^2\cos(l\pi/n)+\sin(l\pi/n)}{9},
& n=4l+2,\quad l\in\mathbb{N}\ \text{odd},\\[6pt]
\displaystyle \frac{7+\cos(l\pi/n)+r_n^2\sin(l\pi/n)}{9},
& n=4l+2,\quad l\in\mathbb{N}\ \text{even}.
\end{cases}
\end{align}
\end{thm}
\begin{proof}
    The proof is direct consequence of \eqref{Rec_n}. The ref. \cite{heinosaari2022random} provides the maximum success merit of Random Access Codes for even $n$-gon. Substituting $\mathcal{S}_{RAC}(n)$ into \eqref{Rec_n}, we find \eqref{even_rec}.
\end{proof}

\begin{thm}
    For odd $n$-gon, the optimum success merit of REC reads as,
  \begin{align}\label{odd_rec}
&\mathcal{S}^{\max}_{\mathrm{REC}}(n)
= \\[-2pt]
&\begin{cases}
\frac{7+\cos(l\pi/n)+r_{2n}^2\sin(l\pi/n)}{9},
& n=4l+1,\quad l\in\mathbb{N}.,\\[6pt]
\frac{7+\cos((l+1)\pi/n)+r_{2n}^2\sin((l+1)\pi/n)}{9},
& n=4l+3,\quad l\in\mathbb{N}.
\end{cases}
\end{align}
\end{thm}
\begin{proof}
  The proof is similar to the case of even $n$-gon. Substituting $\mathcal{S}_{RAC}(n)$ from \cite{heinosaari2022random} into \eqref{Rec_n}, we find \eqref{odd_rec}.
\end{proof}
In Figure \ref{fig:comparison}, we present the plot of $\mathcal{S}^{\max}_{\mathrm{REC}}(n)$ with respect to $n\leq 31$.

\begin{obs}
    All $n$-gon models are nonclassical as $\mathcal{S}^{\max}_{\mathrm{REC}}(n)>\mathcal{S}^C_{REC}=8/9$\cite{excl}.
\end{obs}
\begin{obs}
    When $n\rightarrow\infty$, $\mathcal{S}^{\max}_{\mathrm{REC}}(n)\rightarrow\frac{7+\sqrt{2}}{9}$, which is same as the success merit of REC using qubit states and measurement \cite{excl}.
\end{obs}
\begin{obs}
    For $n=4$, success merit reaches the value of $1$. It is obvious from \eqref{Rec_n} as the merit of Random Access Codes also reaches its algebraic maximum, i.e., $1$ in $n=4$. 
\end{obs}

\begin{widetext}

\begin{figure*}[t]
    \centering
    \begin{subfigure}[b]{0.48\textwidth}
        \centering
        \includegraphics[width=\linewidth]{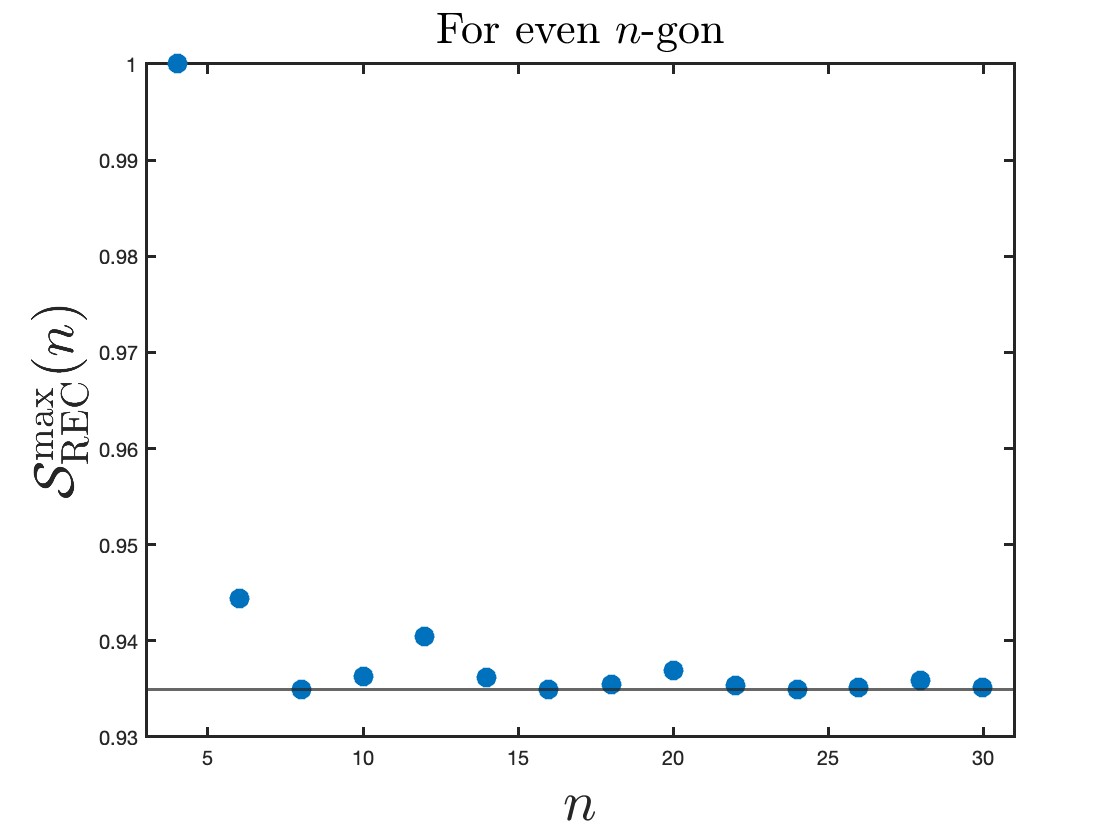}
    \end{subfigure}
    \hfill
    \begin{subfigure}[b]{0.48\textwidth}
        \centering
        \includegraphics[width=\linewidth]{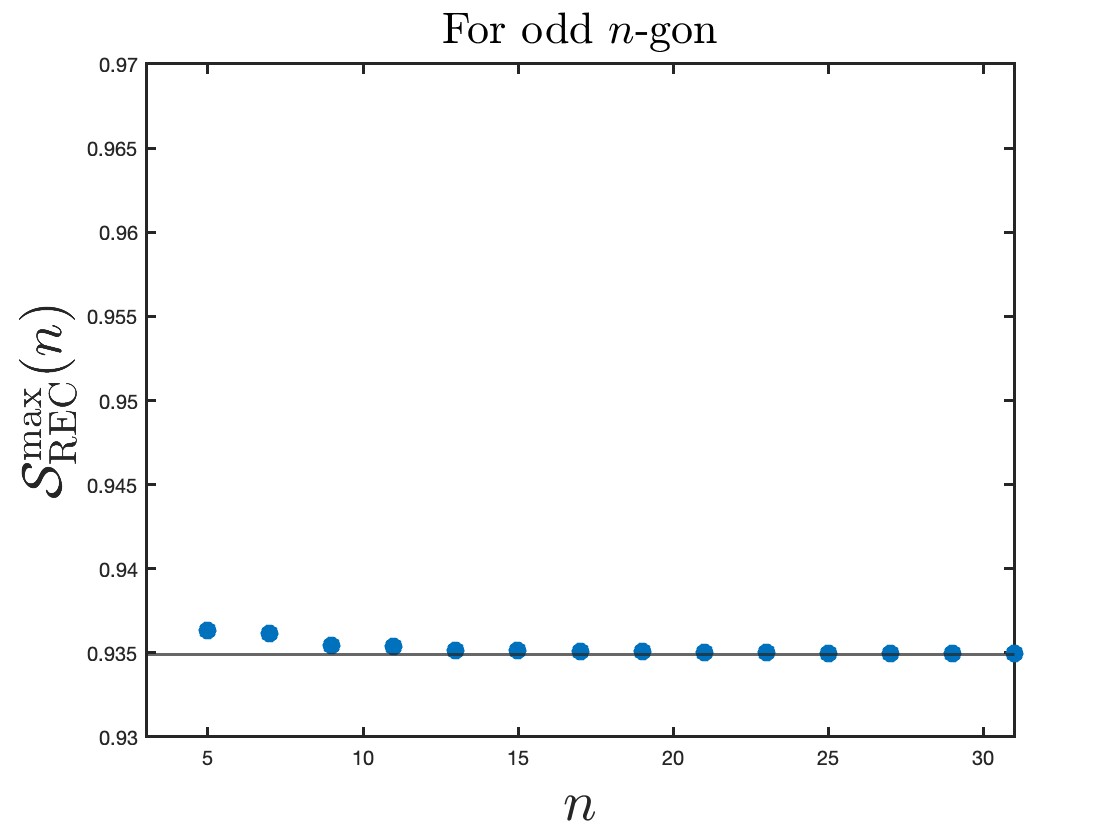}
    \end{subfigure}
    \caption{The optimal success merit of REC for even (left) and odd $n$-gon (right) under dichotomic measurement implemented by Bob. The constant black line corresponds to the optimal success merit of the task using qubit states and measurements. At lower $n$, the success merit has a greater value. Note that the nature of the convergence of the merit is different for even and odd $n$-gon.}
    \label{fig:comparison}
\end{figure*}
\begin{table}[h!]
\centering
\caption{Main results in the polygonal model.}
\label{tab:polygon_results}
\begin{tabular}{lcc}
\hline
\textbf{Property} & \textbf{Even $n$-gon} & \textbf{Odd $n$-gon} \\
\hline
One extremal effect excludes
& $2$ states & $2$ states \\

Separation of excluded states
& $2\pi/n$ & $2\pi/n$ \\

All pure states strongly antidistinguishable
& Yes & Yes \\

Equally spaced $k$-set strongly antidistinguishable
& $k\mid n,\ k\geq 3$ & $k\mid n,\ k\geq 3$ \\

Equally antidistinguishable equally spaced set
& Trine & Trine \\

$n\to\infty$ limit
& Qubit disk & Qubit disk \\
\hline
\end{tabular}
\end{table}
\end{widetext}

\section{Conclusion}
In this work, we developed a systematic framework for studying antidistinguishability in general probabilistic theories. We introduced antidistinguishability, strong antidistinguishability, and equal antidistinguishability, and established general relations among these notions. In particular, we showed that an equally antidistinguishable set of $k$ states in a GPT of affine dimension $d$ must satisfy $k\leq d+1$.
We then investigated antidistinguishability in polygonal theories. The table \ref{tab:polygon_results} summarizes the main results of the corresponding section.
Finally, we introduced Random Exclusion Codes as a communication task for probing nonclassicality in polygonal theories. Our results show that polygon models can outperform the qubit value, with the largest advantages occurring for smaller $n$, while the optimal success merit converges to the qubit value as $n$ increases. The optimal succes merit is considered with dichotomic decoding measurement. For quantum merit, this kind of measurement is proven to be sufficient, for general GPTs, this question remains open.

These results provide a unified perspective on exclusion phenomena in GPTs and motivate several directions for future research. A natural extension is to investigate antidistinguishability beyond polygonal state spaces and in higher-dimensional GPTs. Another promising direction is the study of LOCC antidistinguishability of multipartite states in GPTs. While this question has recently been explored for quantum states \cite{manna20}, its formulation and characterization in general probabilistic theories remain unexplored. It would also be interesting to investigate other information-theoretic tasks based on state exclusion and to examine how their performance varies across different physical theories.

\bibliography{ref}
\end{document}